\documentclass[11pt]{article}
\usepackage{hyperref}
\usepackage{tikz}

\usepackage{amsmath,amsthm,amsfonts,amssymb,amsopn,amscd} 
\usepackage{color}
\def\b{\beta}

\def\e{\varepsilon}

\def\l{\lambda}

\def\A{{\cal A}}

\def\D{{\cal D}}

\def\M{{\cal M}}

\def\R{{\cal R}}

\def\H{{\cal H}}
\def\K{{\cal K}}
\def\S{{\cal S}}

\def\f{{\varphi}}

\def\l{{\lambda}}
\def\x{{{h}}}

\def\PSL{{{\rm PSL}(2,\mathbb R)}}

\def\S2{S^{1(2)}}

\newtheorem{theorem}{Theorem}[section]
\newtheorem{lemma}[theorem]{Lemma}

\newtheorem{corollary}[theorem]{Corollary}

\newtheorem{proposition}[theorem]{Proposition}

\theoremstyle{definition} 

\theoremstyle{remark} 

\newcommand{\ben}{\begin{equation}}
\newcommand{\een}{\end{equation}}

\def\PSL{PSU(1,1)}

\def\SL2{{{\rm SL}(2,\R)}}

\def\PSL2{{{\rm PSL}(2,\Reali)}}

\def\U1{{{\rm V}(1)}}
\def\SU2{{{\rm SV}(2)}}

\def\SU{{{\rm SU}}}

\def\A{{\mathcal A}}

\def\D{{\mathcal D}}

\def\H{{\mathcal H}}

\def\K{{\mathcal K}}

\def\M{{\mathcal M}}

\def\T{{\mathcal T}}

\title{\Huge{The information in a wave
}}

\author{{\sc Fabio Ciolli$^*$, Roberto Longo$^*$, Giuseppe Ruzzi}\thanks{Supported by the ERC Advanced Grant 669240 QUEST ``Quantum Algebraic Structures and Models'', MIUR FARE R16X5RB55W  QUEST-NET and GNAMPA-INdAM. \eject
E-mail: {ciolli@mat.uniroma2.it, longo@mat.uniroma2.it, ruzzi@mat.uniroma2.it}
}
\\
Dipartimento di Matematica,
Universit\`a di Roma Tor Vergata,\\
Via della Ricerca Scientifica, 1, I-00133 Roma, Italy
}

\date{}
\begin{document}

\maketitle

\begin{abstract}
We provide the notion of entropy for a classical Klein-Gordon real wave, that we derive as particular case of a notion entropy for a vector in a Hilbert space with respect to a real linear subspace.
We then consider a localised automorphism on the Rindler spacetime, in the context of a free neutral Quantum Field Theory, that is associated with a second quantised wave, and we explicitly compute its entropy $S$, that turns out to be given by the entropy of the associated classical wave. Here $S$ is defined as the relative entropy between the Rindler vacuum state and the corresponding sector state (coherent state). By $\l$-translating the Rindler spacetime into itself along the upper null horizon, we study the behaviour of the corresponding entropy $S(\l)$. In particular, we show that the QNEC inequality  in the form $\frac{d^2}{d\l^2}S(\l)\geq 0$ holds true for coherent states, because $\frac{d^2}{d\l^2}S(\l)$ is the integral along the space horizon of a manifestly non-negative quantity, the component of the stress-energy tensor in the null upper horizon direction.  
\end{abstract}

\newpage

\section{Introduction} 
Recently, the interplay between Quantum Information and Quantum Field Theory has been subject of deep analyses and new aspects and structures are appearing both in Physics and in Mathematics. 

The physical grounds may ultimately rely on the well know probabilistic nature of Quantum Mechanics, in particular of the wave function. 
Yet, the combination of General Relativity with Quantum Field Theory in the context of Black Hole Thermodynamics showed a fundamental role of Entropy related to Geometry (see \cite{Wa}).
Recently, the study of Entropy in relation with
the Quantum Null Energy Condition has become of much interest (see \cite{CF,Wit} and refs therein). 

Concerning Mathematics, the typical finite dimensional framework of Quantum Information is not sufficient in the QFT context, although it provides essential concepts and results that, possibly, may have an extension to the needed wider infinite dimensional context. The natural language in this context is provided by the Theory of Operator Algebras \cite{T}, in particular by the Tomita-Takesaki modular theory and by Araki's definition of relative entropy for states of a von Neumann algebra \cite{Ar}. 

In the Haag-Kastler approach, QFT is described by the net $\A$ of local von Neumann algebras $\A(O)$ associated to spacetime regions $O$ \cite{H}. In this paper, we are considering the Rindler spacetime $W$ as embedded in the Minkowski spacetime $\mathbb R^{d+1}$, say $W$ is the wedge region $x^1 > |x^0|$ of $\mathbb R^{d+1}$. Motivated by Quantum Information and QNEC aspects, we are going to study the relative entropy $S(\psi|\!|\f)$ between states $\f$ and $\psi$ of the von Neumann algebra $\A(W)$. Here $\f$ is the natural thermal state with respects to the Rindler time-like geodesic flow on $W$, namely the restriction of the Minkowski vacuum state to $\A(W)$, which is KMS with respect to the $x^1$-boost flow at Hawking inverse temperature $2\pi$.

We shall consider the case where the net $\A$ is associated with the free, neutral QFT, and $\psi$ is a coherent state. The first analysis of this kind has been pursued in \cite{L18} in the low dimensional case, namely the chiral conformal QFT associated with the $U(1)$-current. Indeed, in \cite{L18} explicit formulas have been provided also for the relative entropy $S$ of charged states with respect to a half-line  (one-dimensional wedge), and for the first and second derivative $\frac{d}{d\l}S(\l)$ and $\frac{d^2}{d\l^2}S(\l)$ of $S$ as the half-line is translated by $\l$. In particular, the inequality
\[
\frac{d^2}{d\l^2}S(\l) \geq 0
\]
is satisfied with respect the considered states; this inequality (for all states) is the QNEC inequality for constant null shifts, see \cite{CF, LLR}. 

The higher dimensional analysis is more involved, mainly due to tangential contributions on the boundary. The relative entropy for charges localised on the time-zero hyperplane has been studied in \cite{L19} and in \cite{LLR}, see also \cite{CGP} for related analysis. However, these methods were not  
sufficient to compute $S(\l)$ and its first and second derivatives, where $\l$ is a null translation parameter. 

We now illustrate the conceptual steps leading to our results. 
\smallskip

\noindent
{\it The entropy of a vector relative to a real linear subspace.} Our analysis relies on the concept of entropy $S_k$ of a vector $k$ in a Hilbert space $\H$ with respect to a real linear subspace $H$ of $\H$. First of all, we may assume $H$ to be closed, by otherwise considering its closure. Secondly, we may assume that $H$ is a standard subspace of $\H$, by considering only the orthogonal component of $k$ in the standard part of $H$ (see Sect.  \ref{evector}). 
At this point, we may consider the modular operator $\Delta$ and the modular conjugation $J$ of $H$. Finally, we may assume that $H$ is factorial, namely $1$ is not an eigenvalue of $\Delta$, by considering the component of $k$ in the orthogonal of the spectral subspace of the $\Delta$-fixed vectors. So
\[
\overline{H + H' } = \H, \quad H\cap H' = \{0\}\ ,
\]
with $H'$ the symplectic complement of $H$. Our formula for the entropy of $k$ with respect to $H$ is
\ben\label{Skintro}
S_k = - \Im(k, P_H i \log\Delta\, k) \ .
\een
Here $P_H$ is the cutting projection associated with $H$, the projection $P_H : H + H' \to H$, $P_H: h + h' \mapsto h$, that plays a crucial role. We shall give a formula for $P_H$ in terms of $J$ and $\Delta$, 
\[
P_H = \big(a(\Delta) + Jb(\Delta)\big)^-
\]
(closure) with real functions $a$, $b$ given in Sect. \ref{evector}. We shall later compute $P_H$ in the needed context
as one of the ingredients to get our entropy formulas. 
\smallskip

\noindent
{\it The information carried by a classical wave.}
Suppose that information is encoded and transmitted by an undulatory signal, what is the local information carried by the wave packet at later time?

By a wave, we shall here mean a solution of the Klein-Gordon equation
\[
(\square + m^2)\Phi = 0 \ ,
\]
in particular a solution of the wave equation when the mass $m=0$. For definiteness, let's consider real waves  $\Phi$ 
with compactly supported, smooth time zero Cauchy data $\Phi|_{x^0 = 0}$, $\Phi'|_{x^0 = 0}$ that, as is well known, generate and uniquely determine $\Phi$. 
Classical field theory describes $\Phi$ by the stress-energy tensor $T_{\mu\nu}$, which provides the energy-momentum density
of $\Phi$ at any given time. Yet, the information, or entropy, carried by $\Phi$ is not easily visible or definable here. 

Now, the waves' linear space is naturally equipped with a time-independent symplectic form 
\[
\frac12\int_{x^0 =t}(\Phi'\Psi - \Psi'\Phi )dx \ ,
\] 
which is the imaginary part of a Hilbert scalar product, uniquely determined by the relativistic symmetries, so we get a Hilbert space $\H$ by completion, see \cite{J}. Waves with Cauchy data supported in the half-space $x^1 \geq 0$ form a real linear subspace of $\H$, whose closure is denoted by $H(W)$ ($W$ is the causally completed region generated by the half-space). Then we define the entropy $S_\Phi$ of $\Phi$ as the entropy of the vector $\Phi$ with respect to $H(W)$ as above defined. 

By the Bisognano-Wichmann theorem \cite{BW}, the modular group of $H(W)$ corresponds to the rescaled Lorentz boost symmetries in the $x^1$-direction. We shall show that
\ben\label{Sphi}
S_\Phi = 2\pi \int_{x^0 = 0,\, x^{1} \geq 0}x^{1}\, T_{00} \, dx \ ,
\een
with $T_{00} = \frac12(\Phi'^2 + |\nabla\Phi|^2 + m^2\Phi^2)$ the energy density of $\Phi$. 

In order to grasp more on the entropy density of $\Phi$, we have to study the dependence of $S_\Phi$ as $W$ gets translated. We shall give explicit formulas for the null translation case, that motivated our analysis; the  case of a general wedge is similar. 
\smallskip

\noindent
{\it The entropy of a quantum wave-function.}
A quantum, free, neutral, Boson one-particle state is described by a normalised real Klein-Gordon wave $\Phi$. Here however, as is known since the early days of Quantum Mechanics, $\Phi(x^0,{\bf x})$ is not the probability amplitude of finding the particle in position $\bf x$ at time $x^0$ (an approximate description of this probability is given by the Newton-Wigner wave, which is however unsatisfactory because it is not relativistic invariant, see \cite{H}). An intrinsic description of the localisation of $\Phi$ is given in \cite{BGL02}. 

However, a wave $\Phi$ is quantum mechanically localised in a wedge $W$, say $W = \{x: x^1> |x^0|\}$, if the Cauchy data $\Phi_0 , \Phi'_0$ are supported in the time-zero half-space $x^1 \geq 0$ (namely $H(W)$ coincides with the space $H(W)$ as defined in \cite{BGL02}).

In this quantum framework,  the vectors of the above Hilbert space $\H$ represent the particle states. Therefore, also in the QFT context, the entropy of the particle $\Phi$ relative to $W$ is to be given by formula \eqref{Sphi}. 

\smallskip

\noindent
{\it Relative entropy of coherent states in free QFT's.} One way to describe the free scalar QFT is to second quantise the above wave function Hilbert space $\H$ on the Bose Fock Hilbert space, and consider the associated representation of the Weyl algebra. The local von Neumann algebra $\A(O)$ associated with a spacetime region $O$ is the one generated by the Weyl unitaries $V(h)$ with 
$h \in H(O)$, the closure of the real linear subspace of $\H$ of waves localised in $O$. Due to the Weyl commutation relations \eqref{Weyl}, given a vector $k\in \H$, we have an automorphism of $\A(O)$ implemented by $V(k)$; we are particularly interested in the case $O=W$, namely in the automorphism
\[
\b_k = {\rm Ad}V(k)^*\big|_{\A(W)}
\]
which is  normal on the global algebra $\A(W)$ of the Rindler spacetime $W$  associated with the free scalar field. Of course, $\b_k$ acts locally according to the localisation of $k$, namely $\b_k$ acts identically on $\A(O_2)$ if $k$ is localised in $O_1$ and $O_2$ is 
spacelike separated from $O_1$; in other words, $\b_k$ is a localised automorphism.

The entropy $S(\b_k)$ of $\b_k$ (with respect to the vacuum state $\f$) is defined in general \cite{L97}; it is given by the relative entropy 
\[
S(\b_k) = S(\f\cdot \b^{-1}_k |\!| \f) = - (\xi , \log \Delta_{\xi_k, \xi} \xi) \ ;
\]
here $\xi$ is the vacuum vector, $\xi_k = V(k)\xi$ is the vector implementing the state $\f_k = \f\cdot \b^{-1}_k$ on $\A(W)$ and $\Delta_{\xi_k, \xi}$ is the relative modular operator associated with $\A(W)$ and the vectors $\xi, \xi_k$. 

As $\xi_k$ is the coherent vector on the Fock space associated with $k$, we may here equivalently say that $S(\b_k)$ is the relative entropy between the vacuum state $\f$ and the coherent state $\f_k$ on $\A(W)$. 

If $k\in H(W)$, then $\b_k$ is the inner automorphism implemented by $V(k)$; in this case it was shown in \cite{L18,L19} that that $S(\b_k)$ is equal to the entropy of the vector $k$ with respect to $H(W)$ (eq. \eqref{Shh}). We shall show that, in general, we have
\[
S(\b_k) = S_k
\]
with $S_k$ given by \eqref{Skintro}. 
At this point, we have all the ingredients to compute $S(\b_k)$ as the entropy of a classical wave. 

Let $W_\l = W + \l v$ the null translated wedge by the light-like vector $v = (1,1,0\dots, 0)$ (see Section \ref{Rindler} for an intrinsic definition)
and $S_{\Phi}(\l)$ the entropy of the wave $\Phi$, with respect to $H(W_\l)$, thus of the entropy of the associated localised automorphism $\b_\Phi$ of $\A(W_\l)$. 

We shall see that
\[
S_{\Phi}(\l) = 2\pi\int_{x^0 = \l,\,  x^1\geq \l}(x^1 -\l)T_{00}\, dx  \ ,
\]
with $T_{\mu\nu}$ the classical stress energy-momentum tensor of $\Phi$. We then compute that, in particular, 
\[
\frac{d^2}{d\lambda^2} S_{\Phi}(\l)   
= 2\pi \int_{x^0 = \l,  x^1= \l}\langle v, Tv\rangle dx \ ,
\]
so the QNEC inequality $\frac{d^2}{d\lambda^2} S_{\Phi}(\l)\geq 0$ is satisfied for coherent states due to the classical null energy condition.

Note that $\frac{d^2}{d\lambda^2} S_{\Phi}(\l)|_{\l=0}$ depends only on the spatial boundary terms of $\Phi$.

We end up this introduction by summarising in a diagram the logical dependence in our construction
\[
\CD
\boxed{\text {Entropy of a vector}}    @>\text{one-particle space}> \textit{classical field theory}>  \boxed{ \text{Entropy of a wave}} \\ @V\text{Fock }V 
\text{space}\phantom{X} V	@V\text{$2^{nd}$}V \text{quantisation} V  \\  \boxed{\text{Entropy of a coherent state} }   @>\text{local algebras}> \textit{quantum field theory}>     \boxed{\text {Entropy of Klein-Gordon QFT}}\endCD
\]
So the entropy of a vector $\Phi$ with respect to a  real linear subspace $H(O)$ has a double physical interpretation: classically, it measures the information carried by a wave packet in the spacetime region $O$; from the quantum point of view, it gives the vacuum relative entropy, 
on the algebra $\A(O)$ associated with  $H(O)$, of the coherent state induced by $\Phi$ on the Fock space. 

\section{The entropy of a vector}\label{evector}
Our entropy analysis will be based on functional analytic results concerning standard subspaces that we shall obtain in this section and in Section \ref{sect:cs}. We refer to \cite{L, LN} for an account of the results in this context that we will need. 

\subsection{The cutting projection}
Let $\H$ be a complex Hilbert space and $H$ a closed, real linear subspace of $\H$. The subspace $\H_0 =H\cap iH$ 
and $\H_\infty= H'\cap iH'$ are complex Hilbert subspaces and $H$ decomposes according to the direct sum decomposition
\ben\label{oplus}
\H = \H_0\oplus\H_s \oplus\H_\infty \ ,\quad H = H_0\oplus H_s \oplus H_\infty \ ,
\een
with $H_0 = \{0\}$, $H_\infty = \H_\infty$. 

A real linear subspace $H$ of $\H$ is said to be a {\it standard subspace} if $H$ is closed and
\[
\overline{H + iH} = \H\ , \qquad H\cap i H = \{0\}  .
\]
By \eqref{oplus}, every closed real linear subspace $H$ is the direct sum of a standard subspace $H_s$ and two trivial (i.e.\ complex linear) subspaces, and we may deal with standard subspaces only. 

With $H$ a standard subspace, the anti-linear operator $S: H + iH \to H + iH$, $S(h_1 + ih_2) = h_1 - ih_2$ is then  well-defined, closed, involutive. Its polar decomposition $S = J_H\Delta_H^{1/2}$ then gives an anti-linear, involutive unitary and a positive, non-singular, selfadjoint operator $\Delta_H$ on $\H$ such that
\[
\Delta_H^{is} H = H\ , \quad J_H H = H' \ ,
\]
here $H'$ is the symplectic complement $H'$, namely the orthogonal $H' = (iH)^{\bot_\mathbb R}$ of $iH$ with respect to the real scalar product $\Re(\cdot,\cdot)$.
We refer to \cite{T} and \cite{L,LN} for the modular theory of von Neumann algebras and standard subspaces. We shall often denote $\Delta_H$ and $J_H$ simply by $\Delta$ and $J$. 

Let $H\subset \H$ be a standard subspace. 
We say that $H$ is {\it factorial} if
\[
H\cap H' =\{0\}  ;
\]
equivalently, \[
\overline{H + H'} = \H \ .
\]
Thus $H +H'$ is dense in $\H$ and $H+ H'$ is the direct sum (as linear space) of $H$ and $H'$. 

$H$ is factorial iff $1$ is not in the point spectrum of $\Delta$ (see \cite{FG}). Every standard subspace is the direct sum of a factorial standard subspace and an abelian standard subspace, namely
\ben\label{oplus2}
\H = \H_f \oplus \H_a\ , \quad H = H_f \oplus H_a
\een
with $H_f \subset \H_f$ factorial and $\H_a$ the complexification of the real Hilbert space $H_a$. 

In this paper, we shall always assume the standard subspaces to be factorial. However, all results will have an immediate extension to the non-factorial case by the above direct sum decomposition, and indeed to the case of an arbitrary close real linear subspace by \eqref{oplus}.  

Let then $H$ be a standard, factorial subspace of $\H$. 
If $k\in H+H'$, we define 
\[
P_H k = h
\]
where $k = h + h'$ is the unique decomposition of $k$ as a sum of vectors $h\in H$ and $h'\in H'$. This defines a real linear, densely defined map of $\H$ into $\H$, the {\it cutting projection} relative to $H$.  
\begin{lemma}\label{LemmaPH}
We have:
\begin{itemize}
\item[(a)] $P_H$ is a closed, real linear, densely defined operator\ ;

\item[(b)] $P_H^2 = P_H$\ ;

\item[(c)] $P_H + P_{H'} = 1|_{H+H'}$\ ;

\item[(d)] $P^*_H = P_{iH} = -iP_H i$ (adjoint with respect to the real scalar product $\Re (\cdot , \cdot)$) \ ;

\item[(e)] $\Delta_H^{is} P_H = P_H \Delta_H^{is}$ \ .
\end{itemize}
\end{lemma}
\begin{proof}
$(a)$: Let $k_n\in H + H'$ be a sequence such that $k_n \to k$ and $P_H k_n \to h$. Thus $k_n = h_n + h'_n$ with $h_n\in H$,
$h'_n\in H'$, and $h_n + h'_n \to k$ and $h_n\to h$, Hence $h\in H$ and $h'_n \to h'= k - h \in H'$, namely $k\in {\rm Dom}(P_H)$ and $P_H k = h$. 

$(b)$ and $(c)$ are obvious. 

$(d)$:  As $P_H^2 = P_H$, also ${P_H^*}^2 = P^*_H$. 
Then $P_H^* = P_{iH}$ because 
\[
{\rm Ran}(P_H^*)^- = {\rm ker}(P_H)^{\bot_{\mathbb R}} = {H'}^{\bot_{\mathbb R}} = iH\ , \qquad
{\rm ker}(P_H^*) = {\rm Ran}(P_H)^{\bot_{\mathbb R}} = {H}^{\bot_{\mathbb R}} = iH' \ .
\]

$(e)$ follows because $\Delta_H^{is}H = H$, $\Delta_H^{is}H' = H'$. 
\end{proof}
As a consequence of $(e)$, or of Theorem \ref{PH}, for any Borel function $f$ on $(0,+\infty)$, we have the equality
\[
 f(\Delta)P_H k =  P_H \bar f(\Delta^{-1}) k
\]
if $k$ belongs to the domain of (the closure of) both   $f(\Delta)P_H $ and $P_H \bar f(\Delta^{-1})$. 
\subsubsection{A formula for $P_H$} 
Let $H$ be a factorial standard subspace of $\H$ as above. With $0<\e<1/2$, denote by $E_\e$ the spectral projection of $\Delta$ relative to the subset $(\e, 1 -\e) \cup ((1 - \e)^{-1} , \e^{-1})$ of $\mathbb R$, and $\H_\e = E_\e \H$ the corresponding spectral subspace. 

By Lemma \ref{LemmaPH}, we have
\ben\label{EeH}
E_\e H \subset H\ , \quad E_\e H' \subset H' \ .
\een
We also consider the dense, complex linear subspace of $\H$ given by
\[
\D_0 \equiv \bigcup_{\e>0} \H_\e \ .
\]
Let us consider the functions
\[
a(\l) = \l^{-1/2}(\l^{-1/2} - \l^{1/2})^{-1}\ , \quad b(\l) = (\l^{-1/2} - \l^{1/2})^{-1}\ , \ \l\in (0,+\infty)\setminus\{1\}
\]
and set
\[
\D \equiv {\rm Dom}(a(\Delta))\cap {\rm Dom}(b(\Delta))\ .
\]
Clearly, $\D_0\subset \D$ and $\D_0$ is a core for both $a(\Delta)$ and 
$b(\Delta)$. Since the domain of the sum of two operators is the intersection of their domains, $\D$ is the domain of the operator given by the right hand side of \eqref{P1}. 
\begin{theorem}\label{PH}
We have $\D\subset H + H'$ and
\ben\label{P1}
P_H \big|_\D= \Delta^{-1/2}(\Delta^{-1/2} - \Delta^{1/2})^{-1} + J (\Delta^{-1/2} - \Delta^{1/2})^{-1}\ .
\een
Moreover, $\D$ is a core for $P_H$, namely
\[
P_H = \big(a(\Delta) + J b(\Delta)\big)^{-}
\]
where the bar denotes the closure. Indeed, already $\D_0$ is a core for $P_H$. 
\end{theorem}
\begin{proof}
First we assume that $0,1\notin {\rm sp}(\Delta)$, thus sp$(\Delta)$ is bounded. Then both $\D$ and $H + H'$ are equal to $\H$. So any $k\in \H$ can be written as $k = h + h'$, with $h\in H$, $h'\in H'$. Since $J \Delta^{1/2} h = h$ and $J \Delta^{-1/2} h' = h'$,
we have
\begin{multline}
\Delta^{-1/2}(\Delta^{-1/2} - \Delta^{1/2})^{-1}h + J (\Delta^{-1/2} - \Delta^{1/2})^{-1}h \\
= \Delta^{-1/2}(\Delta^{-1/2} - \Delta^{1/2})^{-1}h -  (\Delta^{-1/2} - \Delta^{1/2})^{-1}Jh \\
= \Delta^{-1/2}(\Delta^{-1/2} - \Delta^{1/2})^{-1}h -  (\Delta^{-1/2} - \Delta^{1/2})^{-1}\Delta^{1/2} h = h \ ,
\end{multline}
and
\begin{multline}
\Delta^{-1/2}(\Delta^{-1/2} - \Delta^{1/2})^{-1}h' + J (\Delta^{-1/2} - \Delta^{1/2})^{-1}h' \\
= \Delta^{-1/2}(\Delta^{-1/2} - \Delta^{1/2})^{-1}h' -  (\Delta^{-1/2} - \Delta^{1/2})^{-1}Jh' \\
= \Delta^{-1/2}(\Delta^{-1/2} - \Delta^{1/2})^{-1}h' -  (\Delta^{-1/2} - \Delta^{1/2})^{-1}\Delta^{-1/2} h' = 0 \ ,
\end{multline}
thus \eqref{P1} holds true in this case. 

Now, in the general case, we consider the orthogonal decomposition $\H = \H_\e \oplus \K_\e$, where $\K_\e = \H_\e^\bot$ is the complementary spectral subspace of $\Delta$; we then have a corresponding decomposition $H = H_\e \oplus K_\e$. 
As $0, 1\notin {\rm sp}(\Delta |_{\H_\e})$, we conclude that $\D_0\subset H + H'$ and that \eqref{P1} holds true with $\D_0$ instead of $\D$. 

Let now $k\in\D$ and set $k_\e = E_\e k \in \H_\e$. As $\e\to 0$, we have $k_\e \to k$ and, by the spectral theorem, 
\[
a(\Delta)k_\e \to a(\Delta)k \ ,
\]
\[
b(\Delta)k_\e \to b(\Delta)k \ .
\]
As $P_H$ is closed, we have $k\in{\rm Dom}(P_H) = H + H'$ and 
\[
P_H k = a(\Delta)k + Jb(\Delta)k \ ,
\]
that shows \eqref{P1}. 

Let now $k\in H + H'$. As $P_H$ commutes with $E_\e$, we have that $k_\e \to k$ and
$P_H k_\e = P_H E_\e k= E_\e P_H k \to P_H k$, showing that $\D_0$ is a core for $P_H$.  
\end{proof}
Note that $a(\lambda)=\frac{1}{1-\l}$ and $b(\l)=\frac{\l^{1/2}}{1-\l}$ are real, continuous functions on $(0,+\infty)\setminus\{1\}$, with a singularity of order $1$ at $\l =1$, and bounded out of a neighbourhood of $1$. In particular,  $a$ and $b$ go to $0$ with order $1$ and ${1/2}$ respectively, as $\l \to +\infty$. Hence, for any continuous function $f(\l)$ on $(0,\infty)$ with a zero of order $\geq 1$, at $\l =1$, and that goes to $\infty$   with order  $\leq{1/2}$ as $\l\to \infty$, we have the following corollary (that has an obvious extension to the case of a  Borel function $f$).
\begin{corollary}\label{kks}
Let $f: (0,+\infty) \to\mathbb C$ be a continuous function with the above growth condition. Then ${\rm Ran}(f(\Delta))\subset H+H'$. In particular, for any $k\in\H$, we have $k - \Delta^{is}k\in H + H'$, $s\in\mathbb R$, and
\[
\log\Delta\, k \in H + H', \quad k\in {\rm Dom}(\log\Delta) \ .
\]
\end{corollary}
\begin{proof}
Let $k\in {\rm Dom}(f(\Delta))$. We have
\[
||P_H f(\Delta) k || \leq ||a(\Delta) f(\Delta)  k || + ||b(\Delta)  f(\Delta)  k ||
\]
and $||a(\Delta)  f(\Delta)  k ||,\,||b(\Delta)  f(\Delta)  k || < \infty$; indeed, given $0< \e <1$,  the products $af$ and $bf$ are bounded functions on $(1-\e,1)\cup(1,1+\e)$ and dominated by const.$|f|$ on $(0, 1-\e)\cup(1+\e, \infty)$. 
\end{proof}
\subsection{Definition of entropy of a vector}
Let $H$ be a standard subspace of $\H$. If $h\in H$, the entropy $S_h$ of $h$ w.r.t.\ $H$ is defined in \cite{L18,L19} by
\ben\label{Shh}
S_h = - (h, \log\Delta_H h) \ .
\een
Here, we want to extend this definition for arbitrary vectors in $\H$. 

Let $E(\l)$ be the spectral family of $\Delta_H$, namely
\[
\Delta_H = \int_0^{+\infty}\l dE(\l) \ .
\]
If $k\in \H$, we define the {\it entropy $S_k$ of the vector $k$} with respect to $H$ by
\ben\label{Sdef}
S_k = \Im(k, P_H A k)  \ ,
\een
where $A  =  -i\log\Delta_H$. Indeed $S_k$ takes meaning for any $k$ by defining
\[
-(k, P_H A k) \equiv  i\int_0^{+\infty}  a(\l)\log\l\, d(k,E(\l)k) - i\int_0^{+\infty}  b(\l)\log\l \, d(k,JE(\l)k) \ .
\]
In order to see that the above formula is well defined, note
firstly that $b(\l)\log\l$ is a bounded function, so the right integral is always finite. Secondly, $a(\l)\log\l$ is bounded on $(1 , +\infty)$ and positive on $(0,1)$. 

The entropy $S_k$ is thus finite if\, $-\!\int_0^1 a(\l)\log\l \, d(k,E(\l)k) < \infty$, otherwise $S_k = +\infty$. 
Indeed: 
\begin{proposition}\label{finS}
Let $k\in\H$. We have
\[
S_k < +\infty\quad {\rm iff}\quad  -\!\int_0^1\log\l \, d(k,E(\l)k) < +\infty \ ,
\]
namely iff $k\in {\rm Dom}(\sqrt{|\log \Delta|}\, E_-)$, with $E_-$ the negative spectral projection of $\log\Delta$. 

In particular, all vectors in ${\rm Dom}(\log\Delta)$ have finite entropy. 
\end{proposition}
\begin{proof}
The equivalence follows by the above discussion once we note that both functions $a(\l)\log\l $ and $\log\l$ have a finite limit as $\l \to 1^-$, and that $a(\l) \to 1$ as $\l \to 0^+$. 

The last assertion in the statement is now a clear consequence (it also follows from Corollary \ref{kks}). 
\end{proof}
\noindent
We shall also set $S^H_k= S_k$ if we need to specify the reference standard subspace $H$. 

If $H$ is any closed, real linear subspace of $H$, then the entropy $S^H_k$ of a vector $k\in \H$ with respect to $H$ is defined by
\[
S_k^H = S^{H_f}_{k_f} \,
\]
with $k_f$ the orthogonal component of $k$ in the factorial standard component $H_f$ of $H$ given by \eqref{oplus} and \eqref{oplus2}. 
\begin{proposition}\label{Sproperties}
Let $H$ be a factorial standard subspace of $\H$ and $k\in\H$. The following hold:
\begin{itemize}
\item[$(a)$] If $k = h + h'$, with $h\in H$ and $h'\in H'$,  then $S_k = -(h, \log \Delta\, h)$ (cf. Def.  \eqref{Shh}) ;
\item[$(b)$] $S_k = \Re( k,  iP_H i\log \Delta\, k) = -\Re( k,  P^*_H \log \Delta\, k)$ ;
\item[$(c)$] $S_k \geq 0$ and $S_k = 0$ iff $k\in H'$ ;
\item[$(d)$] If $\H = \H_1\oplus \H_2$ with $H = H_1\oplus H_2$, and $k = k_1\oplus k_2$,
then $S^H_{k}= S^{H_1}_{k_1} + S^{H_2}_{k_2}$, in particular $S^{H_1}_{k_1 } = S^{H}_{k_1 }$ ;
\item[$(e)$] $S_k = \lim_{\e\to 0^+}S_{k_\e}$; here $S_{k_\e} < +\infty$ and $S_{k_\e}$ is non-decreasing as $\e \searrow 0$ ;
\item[$(f)$] $S_k = \lim_{s\to 0} \Re (k, iP_H(\Delta^{is}k - k))/s  = -\lim_{s\to 0} \Im(k, P_H(\Delta^{is}k - k))/s $ ;
\item[$(g)$] If $k_n \to k$ in the graph norm of $\sqrt{|\log \Delta|}\, E_-$, then $S_{k_n} \to S_k$\ .
\end{itemize}
\end{proposition}
\begin{proof}
$(a)$, case $k\in {\rm Dom}(\log\Delta)$, $k\in {\rm Dom}(P_H)$ : In this case $k = h + h'$, with $h\in H$ and $h'\in H'$ and both $h$ and $h'$ in ${\rm Dom}(\log\Delta)$. Thus
\begin{multline}\label{Ska}
S_k =  \Im(k, P_H A k) =  \Im(h+ h', P_H A (h + h'))\\ = \Im(h+ h', A h) =  \Im(h, A h) = -(h, \log\Delta\, h)\ .
\end{multline}
We then have
\ben\label{pos}
S_k \geq 0 \ ;
\een
indeed, in order to check \eqref{pos}, we may assume $||h|| =1$ so, if $E(\l)$ is the spectral family of $\Delta$, we have
\[
(h, \log\Delta\, h) = \int \log(\l)d(h, E(\l) h) \leq  \log\Big(\int d(h, E(\l) h)\Big) = \log||h||^2 = 0\ ,
\]
by Jensen's inequality as $\log$ is a concave function. 

$(d)$ is obvious. 

$(e)$:
$S_{k_\e}$ is finite because the restrictions of $P_H$ and $\log\Delta$ to $\H_\e$ are bounded, and $S_{k_\e}\to S_k$ as $\e\to 0^+$ by the spectral theorem.  Moreover $S_{k_\e}\geq 0$ by \eqref{pos} and Theorem \ref{P1}. Thus $S_{k_\e}$ increases to $S_k$ by $(d)$. 

$(c)$: $S_k \geq 0$ by $(e)$ and the positivity of $S_{k_\e}$ \eqref{pos}. Suppose $S_k = 0$, then $S_{k_\e}= 0$. By $(a)$ and the strict positivity in \eqref{Shh}, we then have $k_\e\in H'$, thus $k\in H'$. 

$(a)$, general case, now follows as $S_k = \lim_{\e \to 0^+}S_{k_\e}$. 

$(b)$: 
Immediate from the definition \eqref{Sdef} and $(d)$ of Lemma \ref{LemmaPH}. Note that 
$(h,  iP_H i\log \Delta\, h) = -(h, \log \Delta\, h) $ is real if $h\in H$. 

$(f)$: We assume that $S_k$ is finite.
By Lemma \ref{finS}, it suffices to show that
\[
\lim_{s\to 0} i(k, E_-\, a(\Delta)(\Delta^{is}k - k))/s = -(k, E_- \,a(\Delta)\log\Delta\, k) \ ,
\]
with $E_-$ the negative spectral projection of $\log\Delta$, thus that
\[
\lim_{s\to 0} \int_0^1 i\frac{(1-\l^{is})}{s} a(\l)  d(k, E(\l)k) =  \int_0^1 \log\l\, a(\l)d(k, E(\l)k) \ ,
\]
that can be checked in Fourier transform as the measure on $\mu$ on $\mathbb R$ given by $d\mu(x)= a(e^x)d(k, E(e^x)k)$ on $(-\infty,0)$ and zero on $[0,+\infty)$ is finite and also $\int x d\mu(x)$ is finite, so $\hat\mu$ is differentiable. The case of infinite $S_k$ follows by similar reasoning by using Fatou's lemma. 

$(g)$ follows by Proposition \ref{finS} and preceding comments. 
\end{proof}

\section{Entropy of Klein-Gordon waves}
We now apply the results in the previous section to compute the entropy of a wave, a solution of the Klein-Gordon equation as will be soon defined, relative to a wedge region $W$, namely the causally complete spacetime region associated with a half-space at a fixed time. 

As is well known, the space of waves carries a natural symplectic form and is embedded in the quantum one-particle Hilbert space $\H$. For every spacetime region $O$, we then have the standard subspace $H(O)$ generated by particle vector states localised in $O$, and we are aimed to compute the entropy of a wave vector with respect to $H(W)$. Here $W$ is a wedge region in the Minkowski spacetime, we choose a fixed wedge
\[
W = \{x\in\mathbb R^{d+1}: x^1 > |x^0|\} \ ;
\]
all other wedges are Poincar\'e translated of $W$. 

The Bisognano-Wichmann theorem \cite{BW} gives 
\ben\label{LW}
\Delta_{H(W)}^{-is} = U(L_W(2\pi s)) \ ,
\een
where $U$ is the unitary representation of the Poincar\'e group on $\H$ and $L_W$ is the one-parameter group of  boosts that preserve $W$. As a consequence
\[
H(W)' = H(W') \ ,
\]
where $W'$ is the spacelike complement of $W$. 

We remark that, in this paper, we use standard tensor notation. We let   
$g$ be the Lorentz metric, i.e.\ the diagonal matrix with signature $1, -1, \dots ,-1$. We denote the Lorentz product $\sum^d_{\mu=0} x^\mu x_\nu$ by $x\cdot x$;  
the partial derivative with respect to $x^\mu$ by $\partial_\mu$ and the partial derivative 
with respect to $x_\mu$ by $\partial^\mu$. Finally we use the relation 
$\partial^\mu = \sum_{\nu = 0}^d g^{\mu\nu}\partial_\nu$.

\subsection{Explicit expression for the entropy}
We now compute explicitly the entropy of a QFT particle state in the case of the quantum free neutral Boson field given by real (generalized) solutions of the Klein-Gordon equation, see \cite{J}. 

\smallskip

Let $\Phi \in S'(\mathbb R^{d+1})$ be a solution of the Klein-Gordon equation, namely $\Phi$ is a tempered distribution and
\[
(\square + m^2)\Phi =0
\]
where $\square = \partial_0^2 - \partial_1^2 - \dots - \partial_d^2$. Then, the Fourier transform $\hat \Phi$ is supported in $\mathfrak H_m  \cup -\mathfrak H_m $, where $\mathfrak H_m =\{p: p\cdot p + m^2 = 0, {p}_{0}\geq 0\}$ is the positive Lorentz hyperboloid. If $\Phi$ is real, then $\hat \Phi(-p) = \overline{\hat \Phi(p)}$, so $\hat \Phi$ is determined by its restriction to $\mathfrak H_m$. 

As is known, if $f, g$ are real, smooth  functions on $\mathbb R^d$ with compact support, there is a unique smooth real solution $\Phi$, on $\mathbb R^{d+1}$,  of the Klein-Gordon equation $(\square + m^2)\Phi = 0$ with Cauchy data $\Phi |_{x^0 =0} = f$, $\Phi' |_{x^0 =0} = g$, where the prime denotes the time derivative $\partial_0$. It can be easily seen that 
$\hat\Phi|_{{\mathfrak H}_m} \in L^2( {\mathfrak H}_m , d\Omega_m)$.

We shall denote by $\cal T$ the real linear space of all real solutions of the Klein-Gordon equation $\Phi$  that are real smooth functions with compactly supported Cauchy data. 
We shall also say that $\Phi$ is a {\it wave} if $\Phi\in\T$.  

Given $\Phi\in\cal T$, we then have a vector $[\Phi]$ in the Hilbert space $\H= L^2( {\mathfrak H}_m , d\Omega_m)$, 
\ben\label{vector}
[\Phi] = \sqrt{2\pi}\hat\Phi |_{\mathfrak H_m}\ ;
\een
these vectors form a real linear, total subset of $\H\equiv L^2( \mathfrak H_m , d\Omega_m)$.  

Denote by $\Phi^\pm$ the positive/negative frequency wave associated with $\Phi\in\T$, namely $\Phi=\Phi^+ +\Phi^-$ where
$\hat \Phi^\pm = \hat \Phi$ on $\pm\mathfrak H_m$ and $\hat \Phi^\pm = 0$ on $\mp\mathfrak H_m$. As $\Phi$ is real, we have
\[
\hat{\Phi}^-(p) = \overline{\hat \Phi^+(-p)}\ , \quad \Phi \in \cal T \ ,
\]
thus the map $\Phi \mapsto [\Phi]$ is one-to-one. 

The one-particle scalar product between the waves $\Phi,\Psi\in \cal T$, given by 
\[
([\Phi], [\Psi]) = \int_{\mathfrak H_m}\overline{\hat\Phi(p)} \hat\Psi(p) d\Omega_m(p) \ , 
\]
is equal to\[
([\Phi], [\Psi]) = i\int_{x^0 =\l} ( \Psi^+ \partial_0 \bar \Phi^+ - \bar \Phi^+ \partial_0  \Psi^+)d{ x}
\]
and is independent of the choice of the constant time integration hyperplane $x^0=\l$. We have
\[
\overline{([\Phi], [\Psi])} = -\int_{-\mathfrak H_m}{\overline{\hat\Phi(p)}} {\hat\Psi(p)}d\Omega_m(p) 
\]
and
\[
2i\Im ([\Phi], [\Psi]) = ([\Phi], [\Psi]) - \overline{([\Phi], [\Psi])} 
=  \int_{\mathfrak H_m\cup-\mathfrak H_m} \overline{\hat{\Phi}(p)} \hat\Psi(p)d\Omega_m(p) 
\]
so, as $\Phi, \Psi$ are real, we have
\ben\label{Cd}
\Im ([\Phi], [\Psi]) =  \frac12\int_{x^0 =\l} \big( \Psi \Phi' -  \Phi  \Psi'\big)dx \   ,
\een
where we denote also by $\Phi'$ the time derivative $\partial^{0}\Phi$ of $\Phi$. We also  denote by
$\Phi_\l$, $\Phi'_\l$ the associated Cauchy data at time $\l$
\[
\Phi_\l({\bf x}) =  \Phi(x^0, {\bf x})|_{x^0 =\l},\quad  \Phi'_\l({\bf x}) =  \partial_{0}\Phi(x^0, {\bf x})|_{x^0 = \l}\ .
\]
Let $O\in \mathbb R^{d+1}$ be a spacetime region. We define the closed, real Hilbert subspace $H(O)$ generated by $[f]$ as $f$ runs in the smooth, real functions on $\mathbb R^{d+1}$ compactly supported in $O$. Here $[f]$ denotes the vector 
$\sqrt{2\pi}\hat f |_{{\mathfrak H}_m} \in \H$. 

Later we shall consider more general solutions of the Klein-Gordon equation that give vectors in $\H$ by \eqref{vector}. Denote by $\bar\T$ the space of real, tempered distributions $\Phi \in S'(\mathbb R^{d+1})$ whose Fourier transform $\hat\Phi$ has the form
\[
\hat\Phi(p) = C(p)\delta(p\cdot p - m^2) \ ,
\]
with $C$ a Borel function on ${\mathfrak H}_m \cup -{\mathfrak H}_m$ 
with $\int_{{\mathfrak H}_m}|C(p)|^2 \delta(p\cdot p - m^2)dp < \infty$. As $\Phi$ is real, $C(-p) = \overline{C(p)}$, so $C$ is determined of ${\mathfrak H}_m$. Clearly, the map \eqref{vector} gives a real linear bijection of $\bar\T$ into $\H$. The inverse map takes $C(p)$ on ${\mathfrak H}_m$, extends it to ${\mathfrak H}_m$ by $C(-p) = \overline{C(p)}$ and Fourier anti-transforms it (apart from a $(2\pi)^{-1/2}$ factor). 

If $\Phi\in\bar\T$, the Cauchy data $\Phi_0 = \Phi|_{x^0 =0}$ and $\Phi'_0 =\Phi'|_{x^0 =0}$ are defined by Fourier anti-transform
\[
\Phi_0({\bf x}) = \frac12\Re \int C(\omega({\bf p}), {\bf p}) e^{i{\bf x}\cdot {\bf p}}\frac{d{\bf p}}{\omega({\bf p})}\ ,
\quad 
\Phi'_0({\bf x}) = \frac12\Im\int C(\omega({\bf p}), {\bf p}) e^{i{\bf x}\cdot {\bf p}}d{\bf p}\ ,
\]  
with $\omega({\bf p}) = \sqrt{{\bf p}^2 + m^2}$. So $\Phi_0$ and $\Phi'_0$ naturally give vectors in $\H$. 
\begin{proposition}\label{supp}
If $\Phi\in\bar\T$, $[\Phi]$ belongs to $H(W)$ if and only if both ${\rm supp}(\Phi_0)$ and ${\rm supp}(\Phi_0')$ are contained in half-space $x^1 \geq 0$. 
\end{proposition}
\begin{proof}
Let $f$ be a real, smooth function on $\mathbb R^d$ with ${\rm supp}(f)\subset W'$. With $D\in S'(\mathbb R^{d+1})$ the commutator function, namely $\hat D(p) = \pm\delta(p\cdot p - m^2)$ if $\pm p_0 \geq 0$,  the convolution $f \star D\in \cal T$ and $[f \star D]\in H(W')$. As $\hat D(p) = 0$ if $p\cdot p <0$, the Cauchy data $(f \star D)_0$, $(f \star D)'_0$ are supported in $x^1 \leq 0$. As vectors $ [f \star D]$ of this form are dense in $H(W')$, it follows that $[\Phi]\in\cal T$ belongs to $H(W)$ iff $ \Im([\Phi], [\Psi]) = 0$ for all such $\Psi$, thus iff the Cauchy data of $\Phi$ are supported in $x^1\geq 0$ by \eqref{Cd}. 
\end{proof}
 Therefore
\[
A[\Phi] = 2\pi [\partial_W \Phi] \ , \quad \Phi\in\T \ .
\]
Here $A=  - i\log\Delta_{H(W)} $ and $\partial_W = x^0\partial_1 + x^1\partial_0$. 

By Prop. \ref{finS}, every $\Phi\in\T$ has finite entropy w.r.t.\ $H(W)$ because $[\Phi]\in {\rm Dom}(\log\Delta_{H(W)})$. 

If $F$ is a function on $\mathbb R^d$, we denote by $F^+$ the function $F^+({\bf x}) = F({\bf x})$ if $x^1 \geq 0$ and $F^+({\bf x}) =0$ if $x^1 <0$. 
\begin{proposition}\label{P+}
Let $\Phi\in\cal T$ and $H = H(W)$ and set $\Psi = \partial_W\Phi$. Then $[\Psi] \in H + H'$ and
$P_H [\Psi]$ is the vector in $\H$ corresponding to the wave with Cauchy data $\Psi^+_{0}$, ${\Psi'_{0}}^+$.
\end{proposition}
\begin{proof}
We have $2\pi[\partial_W\Phi] = i\log\Delta [\Phi]$, so $[\Psi]$ is in the domain of $\log\Delta$ and $[\Psi] \in H + H'$ by Cor.  \ref{kks}. By Prop. \ref{supp} we then have $\Psi = \Psi^+ + \Psi^-$ with $\Psi^\pm \in \bar\T$, $\Psi^+\in H(W)$, $\Psi^-\in H(W')$. We have $P_H [\Psi] = [\Psi^+]$. As discussed, $\Psi^\pm_0$;  we have ${\Psi_0^\pm}'$ give vectors in $\H$ and 
$\Psi_0 = \Psi^+_0 + \Psi^-_0$  and $\Psi'_0 = {\Psi_0^+}' + {\Psi_0^-}' $. 
\end{proof}
If $\Phi\in\cal T$, we shall denote the entropy $S_{[\Phi]}$ of the vector $[\Phi]$ simply by $S_\Phi$. 
\begin{lemma}\label{add}
Let $\Phi\in\T$ be a real, smooth wave, $(\square + m^2)\Phi  = 0$ with Cauchy data $\Phi_0 = f,\Phi'_0 = g$. Then
\[
S_{\Phi} = S_\Gamma + S_\Lambda \ ,
\]
where $\Gamma$ and $\Lambda$ denote the waves with Cauchy data respectively 
$\Gamma_0 = f$, $\Gamma'_0 = 0$ and $\Lambda_0 = 0$, $\Lambda'_0 = g$. 
\end{lemma}
\begin{proof}
We have
\begin{multline}
\frac{1}{2\pi}S_\Phi = \Im(\Phi,P_H A \Phi) = \Im(\Gamma+ \Lambda,P_H A(\Gamma+\Lambda))\\
= \Im(\Gamma, P_H A \Gamma) + \Im(\Lambda,P_H A \Lambda) + \Im(\Gamma, P_H A \Lambda) + \Im(\Lambda,P_H A \Gamma)\\ = \frac{1}{2\pi}S_\Gamma + \frac{1}{2\pi}S_\Lambda+ \Im(\Gamma,P_H A \Lambda) + \Im(\Lambda,P_H A \Gamma)
\ ,
\end{multline}
so we have to show that $\Im(\Gamma,P_H A  \Lambda) + \Im( \Lambda,P_H A\Gamma)= 0$. Now
\[
(\partial_W\Lambda)' = (x^0\partial_{1}\Lambda+ x^1\partial_{0}\Lambda)' 
= \partial_{1}\Lambda+ x^0\partial_{1}\Lambda' + x^1 \Lambda''  = \partial_{1}\Lambda+ x^0\partial_{1}\Lambda' + x^1 (\nabla^2 - m^2)\Lambda 
\]
thus $(\partial_W\Lambda)' |_{x^0 = 0} = 0$. 
We also have
\[
(\partial_W\Gamma)|_{x^0=0}  = (x^0\partial_{1}\Gamma + x^1\partial_{0}\Gamma)|_{x^0=0}= 0
\]
so
\[
\Im(\Gamma,P_H A \Lambda) = \int_{x^0 = 0, x^1 \geq 0} \big(\Gamma' \partial_W\Lambda- (\partial_W\Lambda)' \Gamma\big) \, dx   = 0
\]
and similarly $\Im(\Lambda,P_H A \Gamma) = 0$. 
\end{proof}
With $\l\geq 0$, we shall denote by $W_\l$ the null translated wedge $W_\l = W + \l v$, with $v = (1,1, 0\cdots,0)$. Then
\begin{proposition}
Let $\Phi\in S(\mathbb R^{d+1})$ be a real, smooth wave, $(\square + m^2)\Phi  = 0$. The entropy of $[\Phi]$ w.r.t.\ $H(W_\l)$ is given by
\ben\label{maine}
S_{\Phi}(\l) = \pi\int_{x^0 =\l, x^1\geq \l} \big(\Psi\Phi'    - \Phi \Psi' \big)dx \  ,
\een
with $\Psi =  \partial_{W_\l}\Phi$,  $\partial_{W_\l} =(x^0 -\l)\partial_{1} + (x^1 - \l)\partial_{0}$ and $\Phi' = \partial_{0}\Phi$. 
\end{proposition}
\begin{proof}
Formula \eqref{maine} follows  the definition \eqref{Sdef} and Prop. \ref{P+} if $\l = 0$. The general case is then immediate by translating $\Phi$. 
\end{proof}
We shall consider the {\it energy density} of the wave $\Phi$,
\[
T_{00}(x) = 
\frac12\big(\Phi'^2(x)  + |\nabla \Phi(x)|^2 + m^2\Phi^2(x)\big)
\]
with $\nabla$ the gradient in the space variables (see next section for the full definition of the stress-energy tensor). 
\begin{theorem}\label{maine2}
Let $\Phi\in \T$ be a real Klein-Gordon wave. The entropy $S_{\Phi}(\l)$ of  $[\Phi]$ w.r.t.\ the wedge region $W_\l$ is given by 
\ben\label{mainth}
S_{\Phi}(\l) = 2\pi\int_{x^0 = \l,\,  x^1\geq \l}(x^1 -\l)T_{00}\, dx \ .
\een
\end{theorem}
\begin{proof}
Let's calculate \eqref{maine}. First we assume $\Phi_\l  =0$. Then
\begin{multline}\label{first}
S_{\Phi}(\l) = \pi\int_{x^0 =\l,\, x^1\geq \l}   \Phi' \Psi\, dx \\
 = \pi\int_{x^0 =\l,\, x^1\geq \l}  \Phi' \big((x^0 -\l)\partial_{1}\Phi + (x^1 - \l)\partial_0\Phi\big)dx
= \pi\int_{x^0 =\l,\, x^1\geq \l}  (x^1 - \l)\Phi'^2 \,  dx\ .
\end{multline}
Now we assume instead that $\Phi'_\l =0$. Then
\begin{multline}\label{second}
S_{\Phi}(\l) = -\pi\int_{x^0 =\l,\, x^1\geq \l}    \Phi \Psi' dx  
=-\pi\int_{x^0 =\l,\, x^1\geq \l} \big((x^0 -\l) \partial_{1}\Phi + (x^1-\l)\partial_0\Phi\big)'\, dx
\\
=-\pi\int_{x^0 =\l,\, x^1\geq \l}  \Phi  \partial_{1}\Phi\, dx 
-\pi\int_{x^0 =\l,\, x^1\geq \l}(x^1 -\l)  \Phi \,  \Phi'' \, dx
\\
=\frac{\pi}2\int_{x^0 =\l,\, x^1= \l}  \Phi^2 \, dx 
- \pi\int_{x^0 =\l,\, x^1\geq \l}(x^{1} -\l)  \Phi \, (\nabla^2 - m^2)\Phi\, dx\\
= \pi\int_{x^0 =\l,\, x^{1}\geq \l}(x^1 -\l)  \big(|\nabla\Phi |^2 + m^2\Phi^2\big)dx \ .
\end{multline}
The theorem then follows by Lemma \ref{add} and \eqref{first}, \eqref{second}. 
\end{proof}

\subsection{Entropy, energy-momentum and QNEC inequality}
\label{KG:b}

Let $\Phi\in\T$, namely $\Phi$ is a real smooth solution of the Klein-Gordon equation with compactly supported Cauchy data. 
The classical stress-energy density tensor $T = (T_{\mu\nu})$ associated with $\Phi$ is given by
\[
T_{\mu\nu} =  \partial_\mu\Phi\, \partial_\nu\Phi - g_{\mu\nu}\cal L \ ,
\]
where ${\cal L} = \frac12 \big(\sum_{\nu = 0}^d  \partial_\nu\Phi\, \partial^\nu\Phi - m^2 \Phi^2\big)$
is the Lagrangian density. In particular 
\[
T_{00} = \frac12 \big( \sum_{\mu=0}^d (\partial_\mu \Phi)^2 + m^2\Phi^2   \big)\ ,
\]
\[
T_{0l} =  \partial_0\Phi\, \partial_l\Phi \ , \quad l = 1, 2,... ,d \ .
\]
$T_{0l} $ is the  energy flux across the $x^l = 0$ time-zero surface, $l = 1, 2,\dots, d$.

With $v\in\mathbb R^{d+1}$, consider the energy-momentum in the $v$-direction
\[
\langle v, Tv\rangle \equiv \sum^d_{\mu,\nu = 0}T_{\mu\nu}v^\nu v^\mu \ .
\]
Clearly $T_{00} = \langle e, T e\rangle \geq 0$, 
where $e$ the time-like vector $e = (1,0,\dots 0)$. By Lorentz covariance, we then have
\ben\label{Tv}
\langle v, Tv\rangle \geq 0 
\een
for all time-like vectors $v$, hence for all 
light-like vectors $v$ (classical null energy condition).

The continuity equation holds
\ben\label{cont}
\sum_{\mu=0}^d \partial^\mu T_{\mu\nu} = 0 \ ,
\een
therefore
\[
\int_{x^0 =0,\, x^1 \geq 0} \partial^0 T_{0\nu}\, dx  = -\sum^d_{k=1}\int_{x^0 =0,\, x^1 \geq 0} \partial^k T_{k\nu}\, dx=-\int_{x^0 =0,\, x^1 = 0} T_{1 \nu}\, dx \ .
\] 
\begin{theorem}\label{Ew}
Let $\Phi\in\T$. We have:
\[
\frac{d}{d\lambda} S_{\Phi}(\l)  =
-2\pi\int_{x^0 = \l,  x^1\geq \l}\left(T_{00} + T_{10}\right) \,  dx    \ ,
\]
\[
\frac{d^2}{d\lambda^2} S_{\Phi}(\l)   
= 2\pi \int_{x^0 = \l,  x^1= \l}\langle v, Tv\rangle dx \ ,
\]
where $v$ is the light-like vector $v = (1,1,0\dots, 0)$.
In particular, the inequality
\[
\frac{d^2}{d\lambda^2} S_{\Phi}(\l)  \geq 0
\]
holds true by \eqref{Tv}. 
\end{theorem}
\begin{proof}
Now
\[
S_{\Phi}(\l) = 2\pi\int_{x^0 = \l,\,  x^1\geq \l}(x^1 -\l)T_{00}\, dx  \ ,
\]
so, by the continuity equation \eqref{cont} and integrating by parts, we get 
\begin{multline*}
\frac{d}{d\lambda} S_{\Phi}(\l) = -2\pi\int_{x^0 = \l,\,  x^1\geq \l}\left(T_{00}  - (x^1 -\l)\partial_0T_{00}\right) dx\\
= -2\pi\int_{x^0 = \l,\,  x^1\geq \l}\left(T_{00} - (x^1 -\l)\partial_1T_{10}\right)   dx =
-2\pi\int_{x^0 = \l,\,  x^1\geq \l}\left(T_{00} + T_{10}\right) \,  dx \ ,
\end{multline*}
because $\int_{x^0 = \l,\,  x^1\geq \l}(x^1 -\l)\partial_l T_{l0} dx = 0$ if $l\geq 2$.  

Concerning the second derivative, by the continuity equation \eqref{cont} we have
\begin{align*}
\frac{d^2}{d\lambda^2}S(\lambda)  & =  2\pi\, \int_{x^0=\l,\, x^1 =\l} \left(T_{00} +  T_{10}\right) \,dx 
- 2\pi\, \int_{x^0 =\l,\, x^1\geq \l} \partial_0(T_{00} + T_{10})\,dx  \\
 & = 2\pi\, \int_{x^0=\l,\, x^1=\l} \left(T_{00} +  T_{10}\right) \,dx 
 +2\pi\, \int_{x^0=\l,\, x^1=\l} (T_{10} + T_{11}) dx  \\
 & = 2\pi\, \int_{x^0=\l,\, x^1=\l} \left(T_{00} +  2T_{10}+T_{11}\right) \,dx \ , 
\end{align*}
and this completes the proof.
\end{proof}
For completeness,  note the  manifestly non negative expression

\[
\frac{d^2}{d\lambda^2}S(\lambda)  =  2\pi\, \int_{x^0 = \l,\, x^1= \l} (\partial_0\Phi + \partial_1\Phi)^2\, dx\, ,  
\]
as $T_{00} + 2T_{10} + T_{11} = (\partial_0\Phi + \partial_1\Phi)^2$.

\section{Entropy of coherent states}\label{sect:cs}
We now move to the QFT framework and begin with an abstract second quantisation analysis. 

\subsection{Preliminaries on relative entropy}
Recall that the relative entropy between two normal, faithful states $\f_1, \f_2$ of a von Neumann algebra $\M$
is given by Araki's formula  \cite{Ar}
\ben\label{S12}
S(\f_1 |\!| \f_2) = - (\xi_1 , \log\Delta_{\xi_2,\xi_1}\xi_1)\ ;
\een
here $\xi_1 , \xi_2$ are any cyclic vector representatives of $\f_1 , \f_2$ on the underlying Hilbert space (that always exist in the standard representation) and $\Delta_{\xi_2,\xi_1}$ is the associated relative modular operator. 
The definition
\eqref{S12} is understood as follows.
If $E(\l)$ the spectral family of $\Delta_{\xi_2,\xi_1}$, then
\ben\label{S12bis}
S(\f_1 |\!| \f_2)  = -\int_0^1 \log\l\, d(\xi_1 , E(\l)\xi_1) -  \int_1^{+\infty} \log\l\, d(\xi_1 , E(\l)\xi_1)  \;
\een
the second integral is always finite as $\log\l < \l^{1/2}$ if $\l> 1$. So \eqref{S12} is well defined by \eqref{S12bis}. 
Moreover $S(\f_1 |\!| \f_2) \geq 0$ by Jensen's inequality as $\log\l$ is a concave function. 
\begin{proposition}\label{pre}
\[
S(\f_1 |\!| \f_2) = i \frac{d}{ds}(\xi_1 , \Delta^{is}_{\xi_2,\xi_1}\xi_1)|_{s=0} =  i \frac{d}{ds}\f_1( (D\f_1 : D\f_2)_s)|_{s=0} \ .
\]
\end{proposition}
\begin{proof}
If $S(\f_1 |\!| \f_2) <\infty$, the proposition is proved in \cite{OP}. So, 
$S(\f_1 |\!| \f_2)  = + \infty$ iff $\int_0^1 -\log\l\, d(\xi_1 , E(\l)\x_1) = +\infty$. As 
\[
\int_0^1 -\log\l \,d(\xi_1 , E(\l)\xi_1) \leq \liminf_{s\to 0^+}   \Re\int_0^1 i\frac{\l^{is}-1 }{s} d(\xi_1 , E(\l)\xi_1)
\]
by Fatou's lemma, we have  $ i \frac{d}{ds}(\xi_1 , \Delta^{is}_{\xi_2,\xi_1}\xi_1)|_{s=0}= +\infty$
if $S(\f_1 |\!| \f_2)= +\infty$. The last equation is immediate by the formula
\[
 (D\f_1 : D\f_2)_s =  \Delta^{is}_{\xi_2,\xi_1}\Delta_{\xi_1}^{-is} \ .
\]
\end{proof}
\subsection{Second quantisation preliminaries}\label{SQ}
Let $\H$ be a complex Hilbert space and  $\Gamma(\H)$  the {\em exponential} of $\H$, i.e.\ the Bosonic Fock space over $\H$ (also denoted by $e^\H$). Thus
\[
\Gamma(\H)\equiv \bigoplus^\infty_{n=0}\H_s^{\otimes^n}\ ,
\]
$\H_0\equiv\mathbb C\xi$ is the one-dimensional Hilbert space a unit vector $\xi$ called the {\em vacuum} vector, and $\H_s^{\otimes^n}$ is the symmetric Hilbert $n$-fold tensor product of $\H$. 

If ${h}\in\H$, we denote by $e^{h}$ the coherent vector of $e^\H$:
\[
e^{h}\equiv \bigoplus^\infty_{n=0} \frac{1}{\sqrt{n!}}(h^{\otimes^n})_s
\]
where the zeroth component of $e^h$ is $\xi$, thus $e^0 = \xi$. 
One may check that
\[
(e^{h} ,e^{k}) = e^{({h},{k})}\ ,
\]
and $\{e^{h},\ {h}\in\H\}$ is a total family  of independent vectors of $\Gamma(\H)$.

If $U$ is a (anti-)unitary on $\H$, the second quantisation unitary $\Gamma(U)$  is the (anti-)unitary on $\Gamma(\H)$ given by $\Gamma(U) |_{\H_s^{\otimes^n}}= U\otimes U \otimes\cdots \otimes U$. 
We have $\Gamma(U) e^{h} = e^{U{h}}$. 

With $h\in \H$, the {\em Weyl unitary} $V(h)$ is determined by
\[
V({h})\frac{e^{k}}{||e^{k}||} = \frac{e^{h +k}}{||e^{h +k}||}
\]
and satisfy the Weyl commutation relations
\begin{equation}\label{Weyl}
V({h} +{k}) = e^{i\Im ({h},{k})}V(h)V({k})\ .
\end{equation}
Note that
\begin{equation}\label{cohvacuum}
V({h})\xi = e^{-\frac12 ({h},{h})}e^{{h}} \ ,
\end{equation}
therefore
\[
(V(k)\xi, V(h)\xi) 
= e^{-\frac12 (||h||^2 + ||k||^2)}(e^k , e^h) 
= e^{-\frac12 (||h||^2 + ||k||^2)}e^{(k,h)} 
\]
and, in particular,
\ben\label{fV}
\f(V({h})) = e^{-\frac12 ||{h}||^2} \ .
\een
where $\f \equiv (\xi,\cdot \xi)$ is the vacuum state. 

Let $H\subset\H$ be a real linear subspace. We put
\[
R(H) \equiv \{V({h}):\ {h}\in H\}'' \ ,
\]
namely $R(H)$ is the von Neumann algebra on $\Gamma(\H)$ given by the weak closure of the linear span of the $V({h})$'s as ${h}$ varies in $H$.

Finally, we recall that, if $H$  a standard subspace, the following hold:
\begin{itemize}
\item If $K$ a dense  subspace of $H$, then $R(K)= R(H)$;
\item $\xi$ is a cyclic and separating vector for $R(H)$;
\item Then the modular unitaries and conjugation associated with $(R(H),\xi)$ are given by
$\Delta^{is}_{R(H)} = \Gamma(\Delta^{is}_{H})$,  $J_{R(H)} = \Gamma(J_{H})$.
\item $R(H') = R(H)'$.
\end{itemize}
\noindent
Here, $\Delta_H$ and $J_H$ are the modular operator and the modular conjugation on $\H$ associated with $H$. 
\subsection{First formula, case $h \in H$} 
Given a standard subspace $H$ and a vector $h\in H$, the entropy $S_h = S^H_h$ of $h$ defined in \eqref{Sdef} is given by
\ben\label{Sh0}
S_h = - (h,\log\Delta_H h) \ .
\een
Indeed
\begin{multline*}
S_h = -\Im (h,P_H i\log\Delta_H h) = -\lim_{\e\to 0^+} \Im  (h_\e,P_H i \log\Delta_H h_\e)\\ = -\lim_{\e\to 0^+}  \Im (h_\e, i \log\Delta_H P_H h_\e) 
=  -\lim_{\e\to 0^+} \Im   (h_\e, i \log\Delta_H h_\e)  = - (h,\log\Delta_H h) \ ,
\end{multline*}
because $P_H$ and $i \log\Delta$ commute on $\mathcal{D}_0$. 
Therefore, in this case, the definition of $S_h$ coincides with the definition in \cite{L19}.

We shall see that the relative entropy between coherent states is given by the entropy of vectors. 

Let $h\in\H$ and $V(h)\xi = e^h/{||e^h||}$ the normalised coherent vector. Suppose $H\subset \H$ is a standard subspace and let $\f_h = (V(h)\xi, \cdot\, V(h)\xi)$ as a state on $R(H)$. Note that
\[
\f_h = \f\cdot{\rm Ad}V(h)^* \big |_{R(H)}\ .
\]
We want to study the relative entropy 
\[
S(\f_h |\!| \f_k)\ , \quad h,k\in H\ ,
\]
between the states $\f_h$ and $\f_k$ of $R(H)$. 
Since
\ben\label{hk}
S(\f_h |\!| \f_k) = S(\f\cdot{\rm Ad}V(-h)V(k) |\!| \f) =  S(\f\cdot{\rm Ad}V(k-h) |\!| \f)
=  S(\f_{k-h} |\!| \f) \ ,
\een
we may restrict our analysis to the case $k=0$, namely $\f_k$ is the vacuum state. 

We have
\ben\label{Shf}
S(\f_h |\!| \f) = - ( \xi , \log\Delta_{V(h)\xi, \xi}\xi)=  - ( V(h)\xi , \log\Delta_{R(H)} V(h)\xi)
\een
with $\Delta_{R(H)}$ the modular operator associated with $(R(H),\xi)$. 
The following proposition holds for all $h\in H$, cf.  \cite{L19}. 
\begin{proposition}\label{PropSh}
Let $h\in H$. If $h \in {\rm Dom}(\log\Delta_H)$, the relative entropy on $R(H)$ between $\f$ and $\f_h$ is given by
\[
S(\f_h |\!| \f) = S_h =  - (h, \log\Delta_H h) \ .
\]
\end{proposition}
\begin{proof}

If $h\in H$, the unitary $V(h)$ belongs to $R(H)$, thus 
\[
\Delta^{is}_{V(h)\xi, \xi} = V(h)^* \Delta_\xi^{is} V(h) \Delta_\xi^{-is}\ .
\]
So we have:
\begin{multline*}
S(\f_h |\!| \f) = i\frac{d}{ds}( V(h)\xi , \Delta^{is}_{R(H)} V(h)\xi)\big |_{s=0}
= i\frac{d}{ds}( V(h)\xi , \Gamma(\Delta^{is}_{H}) V(h)\xi)\big |_{s=0} 
\\  = i e^{- ||h||^2}\frac{d}{ds}e^{(h, {\Delta_H^{is}} h)}\big |_{s=0} 
= i\frac{d}{ds}{(h, {\Delta_H^{is}} h)}\big |_{s=0}
= - {(h, \log\Delta_H h)}\ ,
\end{multline*}
where the first equality follows by equation \eqref{S12} similarly as in the proof of Proposition \ref{pre}. 
\end{proof}
Note that $S(\f_h |\!| \f)$ is real, thus
\ben\label{Scom}
S(\f_h |\!| \f) = - (h, \log\Delta_H h) = i(h, i\log\Delta_H h) = -\Im (h, i\log\Delta_H h)
\een
and $S(\f_h |\!| \f)$ may be computed via the symplectic form $\Im(\cdot ,\cdot)$.

\subsection{Entropy of coherent states, general formula}
Let $k\in\H$ and denote by $\f_k$ and $\f'_k$ the states $(V(k)\xi, \cdot V(k)\xi)$ on $R(H)$ and $R(H')$. 
\begin{lemma}\label{Se}
Let $k\in \H_\e$.  Then $S(\f_k |\!| \f) = S_k$. 
\end{lemma}
\begin{proof}
By \eqref{EeH}, we have $k = h + h'$ with $h \in E_\e H$, $h'\in E_\e H'$. So
\[
\-S_k = \Im(h + h', P_Hi\log\Delta (h + h')) = \Im(h + h', i\log\Delta h ) = \Im(h , i\log\Delta h ) = - S_h \ .
\]
On the other hand
\[
S(\f_k |\!| \f)  = S(\f\cdot{\rm Ad}V(k)^* |\!| \f) = S(\f\cdot{\rm Ad}V(h)^* V(h')^* |\!| \f) = S(\f\cdot{\rm Ad}V(h)^* |\!| \f) \ ,
\]
and the lemma follows by Prop. \ref{PropSh}. 
\end{proof}
We now compute the relative entropy between the states $\f$ and $\f_k$ on $R(H)$. 
\begin{lemma}\label{oS}
If $\H = \H_1\oplus \H_2$ with $H = H_1\oplus H_2$, and $k = k_1\oplus k_2$, then
\[
S(\f_k |\!|\f) =  S(\f_{k_1} |\!|\f) + S(\f_{k_2} |\!|\f) \ ,
\] where $\f_{k_i}$ is the coherent state associated with $k_i$ on $\Gamma(\H_i)$. 
\end{lemma}
\begin{proof}
We have $\Gamma(\H) = \Gamma(\H_1) \otimes\Gamma(\H_2)$, $\f = \f\otimes\f$, 
$\f_k = \f_{k_1}\otimes\f_{k_2}$, where we have denoted by $\f$ the vacuum state also  on $\Gamma(\H_i)$. The lemma thus follows by the additivity of the relative entropy under tensor product. 
\end{proof}
\begin{theorem}\label{Su}
Let $k\in \H$. We have
\[
S(\f_k |\!| \f) =   S_k = - \Im(k,  P_H i\log \Delta\, k)  \ .
\]
\end{theorem}
\begin{proof}
Let $k\in\H$. We have the direct sum decomposition $\H = \H_\e \oplus \K_\e$, and $J, \Delta, P_H$ decomposes accordingly. Thus Lemma \ref{oS} gives 
\[
S(\f_k |\!| \f)  = S(\f_{k_\e} |\!| \f)  +  S(\f_{(1- E_\e)k} |\!| \f)  \geq S(\f_{k_\e} |\!| \f)\ .
\]
As $k_\e \to k$, by the semicontinuity of the relative entropy we have 
\[
S(\f_k |\!| \f)\leq \liminf_{\e \to 0^+} S(\f_{k_\e} |\!| \f)\ ,
\]
thus 
\[
\lim_{\e \to 0^+}S(\f_{k_\e} |\!| \f) = S(\f_k |\!| \f) \ .
\]
Since $S(\f_{k_\e} |\!| \f) = S_{k_\e}$ by Lemma \ref{Se}, we have
\[
S(\f_k |\!| \f)  = \lim_{\e \to 0^+}S_{k_\e} = S_k 
\]
by Prop. \ref{Sproperties} $(e)$. 
\end{proof}
\section{Entropy for free QFT's on Rindler spacetime}\label{Rindler}
We now apply the previous analysis to compute the entropy of localised automorphisms for a free scalar Quantum Field Theory on the Rindler spacetime. As a consequence, we shall have the QNEC inequality for coherent states. 

The Rindler spacetime may be identified with the wedge region $W$ of the Minkowski spacetime, as we shall do. 

We denote by $\A$ the net of von Neumann algebras associated with the free, neutral scalar field. Thus
\[
\A(W) = R(H(W)) \ .
\]
$\A(W)$ is the global von Neumann algebra of the the free, neutral scalar field on the Rindler spacetime. 
Given a wave $\Phi\in\T$ we shall consider the automorphism 
\[
\beta_\Phi = {\rm Ad}V([\Phi])^*|_{\A(W)} \ .
\]
The entropy of the automorphism $\beta_k$ with respect to the vacuum is defined as
\[
S(\b_\Phi) = S(\f_\Phi|_{\A(W)} |\!| \f|_{\A(W)}) \ .
\]
More generally, we shall consider the entropy of $\beta_\Phi$ relative to the subwedge $W_\l = W + (\l,\l, 0,\dots,0)$. 

Note that a subregion $W_\l$ can be intrinsically defined  in Rindler spacetime as a the causally complete, space-translation invariant, region spanned by a positive half timelike geodesic. Namely a subregion of $W$ is one of the $W_\l$'s iff it is space-translation invariant, causally complete and mapped into itself by boosts with positive parameter. Once we fix one of these proper  subregions, and call it $W_1$, the null translated region $W_\l$ is given by $L_W(s)W_1$ with $e^s - 1 = \l$. This defines the null translation parameter $\l$ up to rescaling.

\begin{theorem}\label{maine3}
Let $\Phi\in \T$ be a real Klein-Gordon wave. The entropy of $\beta_\Phi$ with respect to the vacuum is given by
\ben\label{mainth2}
S(\f_\Phi|_{\A(W_\l)} |\!| \f|_{\A(W_\l)})   = S_{\Phi}(\l) = 2\pi\int_{x^0 = \l,  x^1\geq \l}(x^1 -\l)T_{00}(x)dx \ ,
\een
where $S_\Phi(\l)$ is the entropy of the wave $\Phi$ with respect to $H(W_\l)$ (cf. Theorem \ref{maine2}).

In particular, the second derivative of $\frac{d^2}{d\l^2}S_\Phi(\l)$ is given by Theorem \ref{Ew}, so the QNEC inequality 
\[
\frac{d^2}{d\l^2}S_\Phi(\l)\geq 0
\] 
holds true for the coherent state associated with $\Phi$. 
\end{theorem}
\begin{proof}
Immediate by Theorems \ref{Ew} and \ref{Su}.
\end{proof}
Note that $\int_{x^0 = 0,\, x^1 \geq 0}\, x^1 T_{00}\,  dx$ may be interpreted as the energy in the Rindler spacetime $W$ in the state given by $\Phi$; the energy is here the one given by the Rindler Hamiltonian, the generator of the time-like Rindler geodesic flow (Minkowskian boosts).

In the following corollary, we denote a vector in $\H$ by $[\Phi]$ and set $S_\Phi = S_{[\Phi]}$, $\f_\Phi = \f_{[\Phi]}$ consistently with the above notations. 
\begin{corollary}
Let $[\Phi]$ be any vector in $\H$. The map 
\[
\l\in [0,\infty) \to S_\Phi(\l) = S(\f_\Phi|_{\A(W_\l)} |\!| \f|_{\A(W_\l)})\in [0,\infty]
\]
is convex. 
\end{corollary}
\begin{proof}
We may assume that there is a $\l$ with $S_\Phi(\l) < \infty$, otherwise the statement is obvious. Let $\l_0$ be the infimum of all $\l$ with 
$S_\Phi(\l) < \infty$; then $S_\Phi(\l) < \infty$ if $\l\in (\l_0,\infty)$ and we have to show that $S_\Phi(\l)$ is convex in $(\l_0,\infty)$. As the pointwise limit of  convex functions is convex, and $S_\Phi(\l)$ is convex if $\Phi\in\T$ by Theorem \ref{maine3}, 
the corollary will  follow  by Prop. \ref{Sproperties} $(g)$ once we show that there exists a sequence of waves $\Phi_n\in \T$ such that $[\Phi_n] \to [\Phi]$ in the graph norm of $\sqrt{|\log\Delta|}\, E_-$.  
This follows because the real linear space $[\T] = \{[\Phi] : \Phi\in\T\}$ is dense in $\H$ in the graph norm of $\sqrt{|\log\Delta|}\, E_-$ as a consequence of Proposition 1.19 in \cite{GL01}. (Note that $[\T]$ is only real linear, but the proof of  
\cite[Proposition 1.19]{GL01} is valid also with a real linear subspace $\D$ there.)
\end{proof}
\section{Outlook}
As the reader might have noticed, our entropy formulas, that ultimately rely on the entropy formula for a vector in the Hilbert space relative to a real linear subspace, depends only on the symplectic form, the imaginary part of the Hilbert space scalar product. They are therefore meaningful in the more general contexts concerning QFT on a curved spacetime background. We aim to study these situations in a forthcoming work, in particular concerning Schwarzschild spacetime, cf. \cite{HI}. 

One important question left open in this paper is the computation of the relative entropy, in the free scalar QFT context, between the vacuum state and any other normal state on $\A(W)$, see however \cite{LLR}. 
\medskip

\noindent
{\bf Acknowledgements.} 
We acknowledge the MIUR Excellence Department Project awarded to the Department of Mathematics, University of Rome Tor Vergata, CUP E83C18000100006.

\end{document}